\documentclass[11pt]{article}
\usepackage[margin=1in]{geometry}

\usepackage[english]{babel}
\usepackage[utf8x]{inputenc}
\usepackage[T1]{fontenc}

\usepackage{amsmath}
\usepackage{amssymb}
\usepackage{amsthm}
\usepackage{mathtools}
\usepackage{amsfonts}
\usepackage{enumitem}
\usepackage{graphicx}
\usepackage{color}
\usepackage[colorinlistoftodos]{todonotes}
\usepackage[colorlinks=true, allcolors=blue]{hyperref}
\RequirePackage[colorlinks=true]{hyperref}
\hypersetup{
  linkcolor=[rgb]{0,0,0.5},
  citecolor=[rgb]{0, 0.5, 0},
  urlcolor=[rgb]{0.5, 0, 0}
}

\usepackage{algorithmicx}
\usepackage{algorithm}
\usepackage[noend]{algpseudocode}

\algnewcommand\algorithmicinput{\textbf{Input:}}
\algnewcommand\Input{\item[\algorithmicinput]}

\algnewcommand\algorithmicoutput{\textbf{Output:}}
\algnewcommand\Output{\item[\algorithmicoutput]}

\usepackage{tabularx}

\newif\ifblog
\newif\iftex
\blogfalse
\textrue

\DeclarePairedDelimiter\ceil{\lceil}{\rceil}
\DeclarePairedDelimiter\floor{\lfloor}{\rfloor}

\usepackage{ulem}

\def\emph#1{\textit{#1}}

\newcommand{\tB}{\tilde{B}}

\def\E{{\mathbb E}}

\def\N{{\mathbb N}}
\def\Z{{\mathbb Z}}
\def\R{{\mathbb R}}

\def\eps{{\epsilon}}

\def\be{{\bf e}}
\def\bx{{\bf x}}

\usepackage{amsthm}
\usepackage{thmtools,thm-restate}

\numberwithin{equation}{section}

\declaretheorem[numberlike=equation]{theorem}

\declaretheorem[name=Theorem,numbered=no]{theorem*}

\declaretheorem[numberlike=equation]{lemma}
\declaretheorem[name=Lemma,numbered=no]{lemma*}

\declaretheorem[numberlike=equation]{corollary}
\declaretheorem[name=Corollary,numbered=no]{corollary*}

\declaretheorem[name=Claim,numbered=no]{claim*}

\newtheorem{definition}[theorem]{Definition}

\theoremstyle{remark}
\declaretheorem[numberlike=equation]{remark}
\declaretheorem[name=Remark,numbered=no]{remark*}

\title{Tracking the $\ell_2$ Norm with Constant Update Time}
\author{
Chi-Ning Chou\thanks{School of Engineering and Applied Sciences, Harvard University, Cambridge, Massachusetts, USA. Supported by NSF awards CCF 1565264 and CNS 1618026. Email: \texttt{chiningchou@g.harvard.edu}, and~\texttt{leizhixian.research@gmail.com}.}
\and Zhixian Lei\footnotemark[1]
\and Preetum Nakkiran\thanks{School of Engineering and Applied Sciences, Harvard University, Cambridge, Massachusetts, USA. Email: \texttt{preetum@cs.harvard.edu}.
Work supported in part by a Simons Investigator Award, NSF
Awards CCF 1565641 and CCF 1715187, and the NSF Graduate Research Fellowship Grant No. DGE1144152.}
}

\newcommand{\fix}[1]{{ #1}}

\begin{document}
\maketitle

\begin{abstract}
The \emph{$\ell_2$ tracking problem} is the task of obtaining a streaming algorithm that, given access to a stream of items $a_1,a_2,a_3,\ldots$ from a universe $[n]$, outputs at each time $t$ an estimate to the $\ell_2$ norm of the \textit{frequency vector} $f^{(t)}\in \mathbb{R}^n$ (where $f^{(t)}_i$ is the number of occurrences of item $i$ in the stream up to time $t$).
The previous work [Braverman-Chestnut-Ivkin-Nelson-Wang-Woodruff, PODS 2017] gave a streaming algorithm  with (the optimal) space using $O(\epsilon^{-2}\log(1/\delta))$ words and  $O(\epsilon^{-2}\log(1/\delta))$ update time to obtain an $\epsilon$-accurate estimate with probability at least $1-\delta$.
We give the first algorithm that achieves  update time of $O(\log 1/\delta)$ which is independent of the accuracy parameter $\epsilon$, together with the nearly optimal space using $O(\epsilon^{-2}\log(1/\delta))$ words.
Our algorithm is obtained using the \textsf{Count Sketch} of [Charilkar-Chen-Farach-Colton, ICALP 2002].
\end{abstract}

\pagenumbering{arabic}

\section{Introduction}
The \textit{streaming model} considers the following setting. One is given a list $a_1,a_2,\dots,a_m\in[n]$ as input where we think of $n$ as extremely large. The algorithm is only allowed to read the input once in a stream and the goal is to answer some predetermined queries using space of size logarithmic in $n$. For each $i\in[n]$ and time $t\in[m]$, define $f^{(t)}_i=|\{1\leq j\leq t:\ a_j=i\}|$ as the frequency of $i$ at time $t$.
Many classical streaming problems are concerned with approximating statistics of $f^{(m)}$ such as the distinct element problem (\textit{i.e., }$\|f^{(m)}\|_0$). One of the most well-studied problems is the one-shot $\ell_2$ estimation problem where the goal is to estimate $\|f^{(m)}\|_2^2$ within multiplicative error $(1\pm\eps)$ and had been achieved by the seminal \textsf{AMS} sketch by Alon et al.~\cite{alon1996space}.

We consider a streaming algorithm $A$ that maintains some logarithmic space and outputs an estimation $\sigma_t$ at the $t^{\text{th}}$ step of the computation. $A$ achieves $\ell_2$ $(\epsilon,\delta)$-tracking if for every input stream $a_1,a_2,\dots,a_m\in[m]$
\begin{equation*}
\Pr\left[ \exists_{t \in [m]} \bigl| \sigma_t - \| f^{(t)} \|_2^2 \bigr| > \epsilon\Delta_t  \right] \leq \delta
\end{equation*}
where the ``normalization factor'' $\Delta_t$ differs between \textit{strong} tracking and \textit{weak} tracking. For $(\epsilon,\delta)$-\textit{strong tracking}, $\Delta_t = \| f^{(t)}\|_2^2$ is the norm squared of the frequency vector up to the time $t$, while for $(\epsilon,\delta)$-weak tracking, $\Delta_t = \| f^{(m)} \|_2^2$ is the norm squared of the overall frequency vector. Note that strong tracking implies weak tracking and weak tracking implies one-shot approximation.
In this work, we focus on $\ell_2$ tracking via linear sketching, where we specify a distribution $D$ on matrices $\Pi \in \R^{k \times n}$, and maintain a sketch vector at time $t$ as $\tilde f^{(t)} \triangleq \Pi f^{(t)}$. Then the estimate $\sigma_t$ is defined as $\|\tilde{f}^{(t)}\|_2^2$.
The space complexity of $A$ is the number of machine words\footnote{Following convention, we assume the size of a machine word is at least $\Omega(\max(\log n, \log m))$ bits.} required by $A$.
The update time complexity of $A$ is the time to update $\sigma_t$, in terms of number of arithmetic operations.

Both weak tracking and strong tracking have been studied in different context~\cite{huang2014tracking,braverman2016beating,braverman2017bptree} and the focus of this paper is on the \textit{update time complexity}. 
Specifically, we are interested in the dependency of update time on the approximation factor $\epsilon$. The state-of-the-art result prior to our work is by Braverman et. al.~\cite{braverman2017bptree} showing that \textsf{AMS} provides weak tracking with $O(\epsilon^{-2}\log(1/\delta))$ update time and $O(\epsilon^{-2}\log(1/\delta))$ words of space.

Apart from tracking, there have been several sketching algorithms for one-shot approximation that have faster update time. Dasgupta et. al.~\cite{dasgupta2010sparse} and Kane and Nelson~\cite{kane2014sparser} showed that sparse JL achieves $O_\delta(\epsilon^{-1})$~\footnote{$O_\delta(\cdot)$ is the same as the usual big O notation except treating $\delta$ as a constant.} update time for $\ell_2$ one-shot approximation. Charikar, Chen, and Farach-Colton~\cite{charikar2002finding} designed the \textsf{CountSketch} algorithm for the heavy hitter problem and Thorup and Zhang~\cite{thorup2004tabulation} showed that it achieve $O_\delta(1)$ update time for $\ell_2$ one-shot approximation.

\paragraph{Update time}
Unlike the space complexity in streaming model, there have been less studies in the update time complexity though it is of great importance in applications. For example, the \textit{packet passing problem}~\cite{krishnamurthy2003sketch} requires the $\ell_2$ estimation in the streaming model with input arrival rate as high as $7.75\times10^6$ packets~\footnote{Each packet has 40 bytes (320 bits).} per second. Thorup and Zhang~\cite{thorup2012tabulation} improved the update time from 182 nanoseconds to 50 nanoseconds and made the algorithm more practical.

While some streaming problems have algorithms with constant update time
(\textit{e.g.,} distinct elements~\cite{kane2010optimal} and $\ell_2$ estimation~\cite{thorup2012tabulation}),
some other important problems do not ($\ell_p$ estimation for $p\neq2$~\cite{kane2011fast}, heavy hitters problems\footnote{There is a memory and update time tradeoff for heavy hitter from space $O(\epsilon^{-2} \log(n/\delta))$ to $O(\epsilon^{-2}(n/\delta))$ to get constant update time. However, achieving constant update time and logarithmic space simultaneously is unknown.}~\cite{charikar2002finding,cormode2005improved}, and tracking problems~\cite{braverman2017bptree}).
Larsen et al.~\cite{larsen2015time} systematically studies the update time complexity and showed lower bounds against heavy hitters, point query, entropy estimation, and moment estimation in the non-adaptive turnstile streaming model.
In particular, they show that $O(\eps^{-2})$-space algorithms
for $\ell_2$ estimation of vectors over $\R^n$, with failure probability $\delta$,
must have update time roughly $ \Omega(\log(1/\delta)/\sqrt{\log n})$. Note that their lower bound does not depend on $\epsilon$.

\paragraph{Space lower bounds}
For one-shot estimation of the $\ell_2$ norm, Kane et al.~\cite{kane2010exact} showed that
$\Theta(\epsilon^{-2}\log m+\log\log n)$ bits of space are required,
for any streaming algorithm. This space lower bound is tight due to the \textsf{AMS} sketch.
However, this only applies in the constant failure probability regime.

In the regime of sub-constant failure probability $\delta$,
known tight lower-bounds on Distributional JL~\cite{KMN11, JW13}
imply that $\Omega(\eps^{-2}\log(1/\delta))$ rows are necessary
for the special case of linear sketching algorithms.~\footnote{
Note that an $(\eps, \delta)$-weak tracking via linear sketch
defines a distribution over matrices that satisfies the
Distributional JL guarantee, with distortion $(1 \pm \eps)$ and failure probability $\delta$.
}
For linear sketches, this lower bound on number of rows
is equivalent to a lower bound on the words of space.

For the regime of faster update time, Kane and Nelson~\cite{kane2014sparser}
shows that \textsf{CountSketch}-type of constructions (with the optimal $\Omega(\eps^{-2}\log(1/\delta))$ rows) require sparsity i.e. number of non-zero elements
$\widetilde\Omega(\eps^{-1}\log(1/\delta))$~\footnote{$\tilde{\Omega}(\cdot)$ is the same as the $\Omega(\cdot)$ notation by ignoring extra logarithmic factor.} per column to achieve distortion $\eps$ and failure probability $\delta$.
But, this does not preclude a sketch with suboptimal dependency on $\delta$ in the number of rows
from having constant sparsity, for example a sketch with $\Omega_\delta(\eps^{-2})$ rows and constant sparsity -- indeed, this is what \textsf{CountSketch} achieves.
Note that in our setting, we can boost constant-failure probability to arbitrarily small failure probability by taking medians of estimators.\footnote{This is not immediate for weak tracking.}
Thus, we may be able to bypass the lower-bounds for linear sketches.

To summarize the situation:
for constant failure probability, it is only known
that linear sketches require dimension $\Omega(\eps^{-2})$,
and it is not known if super-constant sparsity is required for tracking with this optimal dimension.
In particular, it was not known how to achieve say $(\eps, O(1))$-weak tracking for $\ell_2$, with $O(\eps^{-2})$ words of space and constant update time.

\paragraph{Our contributions}
In this paper, we show that there is a streaming algorithm with $O(\log(1/\delta))$ update time and space using $O(\epsilon^{-2}\log(1/\delta))$ words that achieves $\ell_2$ $(\epsilon,\delta)$-weak tracking. 
\begin{theorem}[informal]\label{thm:informal}
For any $\epsilon>0$, $\delta\in(0,1)$, and $n\in\N$. For any insertion-only stream over $[n]$ with frequencies $f^{(1)},f^{(2)},\dots,f^{(m)}$, there exists a  streaming algorithm providing $\ell_2$ $(\epsilon,\delta)$-weak tracking with space using $O(\epsilon^{-2}\log(1/\delta))$ words and $O(\log(1/\delta))$ update time.
\end{theorem}

Further, by applying a standard union bound argument in~\autoref{lem:weak to strong}, the same algorithm can achieve $\ell_2$ strong tracking as well.
\begin{corollary}
For any $\epsilon>0$, $\delta\in(0,1)$, and $n\in\N$. For any insertion-only stream over $[n]$ with frequencies $f^{(1)},f^{(2)},\dots,f^{(m)}$, there exists a  streaming algorithm providing $\ell_2$ $(\epsilon,\delta)$-strong tracking with $O(\epsilon^{-2}\log(1/\delta)\log\log m)$ words and $O(\log(1/\delta)\log\log m)$ update time.
\end{corollary}

The algorithm in the main theorem is obtained by running $O(\log(1/\delta))$ many copies of \textsf{CountSketch} and taking the median.

The main techniques used in the proof are the chaining argument and Hansen-Wright inequality which are also used in~\cite{braverman2017bptree} to show the tracking properties of \textsf{AMS}. However, direct applications of these tools on the \textsf{CountSketch} algorithm would not give the desired bounds due to the sparse structure of the sketching matrix. To overcome this issue, we have to dig into the structure of sketching matrix of \textsf{CountSketch}. We will compare the difference between our techniques and that in~\cite{braverman2017bptree} after presenting the proof of~\autoref{thm:informal} (see~\autoref{rmk:differences}).\\

The rest of the paper is organized as follows.
Some preliminaries are provided in~\autoref{sec:preliminaries}.
In~\autoref{sec:cs weak tracking},
we prove our main theorem showing that \textsf{CountSketch} with $O(\eps^{-2})$ rows achieves $\ell_2$ $(\epsilon,O(1))$-weak tracking with constant update time.
As for the $\ell_2$ strong tracking, we discuss some upper and lower bounds in~\autoref{sec:strong}. In~\autoref{sec:conclusion}, we discuss some future directions and open problems.

\section{Preliminaries}\label{sec:preliminaries}
In the following, $n\in\N$ denotes the size of the universe, $k$ denotes the number of rows of the sketching matrix, $t$ denotes the time, and $m$ denote the final time. We let $[n]=\{1,2,\dots,n\}$ and use $\tilde{O}(\cdot)$ and $\tilde{\Omega}(\cdot)$ to denote the usual $O(\cdot)$ and $\Omega(\cdot)$ with some extra poly-logarithmic factor.

The input of the streaming algorithm is a list $a_1,a_2,\dots,a_m\in[n]$. For each $i\in[n]$ and time $t\in[m]$, define $f^{(t)}_i=|\{1\leq j\leq t:\ a_j=i\}|$ as the frequency of $i$ at time $t$. The one-shot $\ell_2$ approximation problem is to produce an estimate for $\|f^{(m)}\|_2^2$ with $(1\pm\epsilon)$ multiplicative error and success probability at least $1-\delta$ for $\epsilon>0$ and $\delta\in(0,1)$.

\subsection{$\ell_2$ tracking}
Here, we give the formal definition of $\ell_2$ tracking for sketching algorithm.
\begin{definition}[$\ell_2$ tracking]
For any $\epsilon>0, \delta\in(0,1)$, and $n,m\in\N$. Let $f^{(1)},f^{(2)},\dots,f^{(m)}$ be the frequency of an insertion-only stream over $[n]$ and $\tilde{f}^{(1)},\tilde{f}^{(2)},\dots,\tilde{f}^{(m)}$ be its (randomized) approximation produced by a sketching algorithm. We say the algorithm provides $\ell_2$ $(\epsilon,\delta)$-strong tracking if
$$
\Pr\left[\exists_{t \in [m]},\ \Bigl|\|\tilde{f}^{(t)}\|_2^2-\|f^{(t)}\|_2^2\Bigr|>\epsilon\|f^{(t)}\|_2^2\right]\leq\delta.
$$
We say the algorithm provides $\ell_2$ $(\epsilon,\delta)$-weak tracking if
$$
\Pr\left[\exists_{t \in [m]},\ \Bigl|\|\tilde{f}^{(t)}\|_2^2-\|f^{(t)}\|_2^2\Bigr|>\epsilon\|f^{(m)}\|_2^2\right]\leq\delta.
$$
\end{definition}
Note that the difference between the two tracking guarantee is that in strong tracking we bound the deviation of the estimate from the true norm squared by $\epsilon\|f^{(t)}\|_2^2$ while in the weak tracking we bound this deviation by $\epsilon\|f^{(m)}\|_2^2$.


\subsection{\textsf{AMS} sketch and \textsf{CountSketch}}
Alon \textit{et. al.}~\cite{alon1996space} proposed the seminal \textsf{AMS} sketch for $\ell_2$ approximation in the streaming model.
In \textsf{AMS} sketch, consider $\Pi\in\R^{k\times n}$ where $\Pi_{j,i}=\sigma_{j,i}/\sqrt{k}$ and $\sigma_{j,i}$ is i.i.d.~Rademacher for each $j\in[m],i\in[n]$.
When $k=O(\epsilon^{-2})$, \textsf{AMS} sketch approximates $\ell_2$ norm within $(1\pm\epsilon)$ multiplicative error. Note that the update time of \textsf{AMS} sketch is $k$ since the matrix $\Pi$ is dense.

Charikar, Chen, and Farach-Colton~\cite{charikar2002finding} proposed the following \textsf{CountSketch} algorithm for the heavy hitter problem and Thorup and Zhang~\cite{thorup2004tabulation} showed that \textsf{CountSketch} is also able to solve the $\ell_2$ approximation.
Here, consider $\Pi\in\R^{k\times n}$ where we denote the $i^{\text{th}}$ column of $\Pi$ as $\Pi_i$ for each $i\in[n]$. $\Pi_i$ is defined as follows. First, pick $j\in[k]$ uniformly and set $\Pi_{j,i}$ to be an independent Rademacher. Next, set the other entries in $\Pi_i$ to be 0. Note that unlike \textsf{AMS} sketch, the normalization term in \textsf{CountSketch} is 1 since there is exactly one non-zero entry in each column.~\cite{charikar2002finding} showed that \textsf{CountSketch} provides one-shot $\ell_2$ approximation with $O(\epsilon^{-2})$ rows.
\begin{lemma}[\cite{charikar2002finding,thorup2004tabulation}]\label{lem:CS estimation}
	Let $\epsilon>0$, $\delta\in(0,1)$, and $n\in\N$. Pick $k=\Omega(\epsilon^{-2}\delta^{-1})$, we have for any $x\in\R^n$,
    \begin{equation*}
    \Pr_{\Pi}\left[|\|\Pi \bx\|_2^2-\|\bx\|_2^2|>\epsilon\|\bx\|_2^2\right]\leq\delta.
    \end{equation*}
\end{lemma}

\paragraph{Implement \textsf{CountSketch} in logarithmic space}
Previously, we defined \textsf{CountSketch} using uniformly independent randomness, which requires space $\Omega(nk)$. However, one could see that in the proof of~\autoref{thm:CS weak} we actually only need 8-wise independence. Thus, the space required can be reduced to $O(\log n)$ for each row.
It is well known that \textsf{CountSketch} with $k$ rows can be implemented with 8-wise independent hash family using $O(k)$ words. We describe the whole implementation in~\autoref{sec:CS} for completeness.

\subsection{$\epsilon$-net for insertion-only stream}
In our analysis, we will use the following existence of a small $\epsilon$-net for insertion-only streams.

\begin{definition}[$\eps$-net]
	Let $S\subseteq\R^n$ be a set of vectors. For any $\eps>0$, we say $E\subseteq\R^n$ is an $\eps$-net for $S$ with respect to $\ell_2$ norm if for any $x\in S$, there exists $y\in E$ such that $\|x-y\|_2\leq\eps$.
\end{definition}

\begin{lemma}[\cite{braverman2016beating}]\label{lem:eps net insertion-only stream}
Let $\{x^{(t)}\}_{t\in[m]}$ be an insertion-only stream. For any $\epsilon>0$, there exists a size $\left(1+\epsilon^{-2}\cdot\|x^{(m)}\|_2\right)$ $\epsilon$-net for $\{x^{(t)}\}_{t\in[m]}$ with respect to $\ell_2$ norm. Moreover, the elements in the net are all from $\{x^{(t)}\}_{t\in[m]}$.  
\end{lemma}
\begin{proof}[Proof Sketch]
	The idea is to use a greedy algorithm, by scanning through the stream from the beginning and adding an element $x^{(t)}$ into the net if there does not already exist an element in the net that is $\epsilon$-close to $x^{(t)}$.
\end{proof}

\subsection{Concentration inequalities}
Our analysis crucially relies on the following Hanson-Wright inequality~\cite{hanson1971bound}.
\begin{lemma}[Hanson-Wright inequality~\cite{hanson1971bound}]\label{lem:HW}
For any symmetric $B\in\R^{n\times n}$, $\sigma\in\{\pm1\}^n$ being independent Rademacher vector, and integer $p\geq1$, we have
\begin{equation*}
\|\sigma^\top B\sigma-\E_\sigma[\sigma^\top B\sigma]\|_p\leq O\left(\sqrt{p}\|B\|_F+p\|B\|\right)=O(p\|B\|_F),
\end{equation*}
where $\|X\|_p$ is defined as $\E[|X|^p]^{1/p}$ and $\|\cdot\|_F$ is the Frobenius norm.
\end{lemma}
Note that the only randomness in $\sigma^\top B\sigma-\E_\sigma[\sigma^\top B\sigma]$ is the Rademacher vector $\sigma$.

\section{\textsf{CountSketch} with $O(\epsilon^{-2})$ rows provides $\ell_2$ weak tracking}\label{sec:cs weak tracking}
In this section we will show that \textsf{CountSketch} with $O(\epsilon^{-2})$ rows provides $(\eps, O(1))$-weak tracking.

\begin{theorem}[\textsf{CountSketch} with $O(\epsilon^{-2})$ rows provides $\ell_2$ weak tracking]\label{thm:CS weak}
For any $\epsilon>0$, $\delta\in(0,1)$, and $n\in\N$. Pick $k=\Omega(\epsilon^{-2}\delta^{-1})$. For any insertion-only stream over $[n]$ with frequency $f^{(1)},f^{(2)},\dots,f^{(m)}$, the \textsf{CountSketch} algorithm with $k$ rows provides $\ell_2$ $(\epsilon,\delta)$-weak tracking.
\end{theorem}

\begin{remark*}
Note that for linear sketches, the dependency of number of rows on $\epsilon$ is tight in~\autoref{thm:CS weak}.
This is implied by known lower-bounds on Distributional JL~\cite{KMN11, JW13}, which imply lower-bounds on one-shot $\ell_2$ approximation.
\end{remark*}

\begin{remark*}
Recall that the number of rows in linear sketches is proportional to the number of words needed in the algorithm.
\end{remark*}

Using the standard median trick, we can run $O(\log(1/\delta))$ copies of \textsf{CountSketch} with $k=O(\epsilon^{-2})$ in parallel and output the median. With this, ~\autoref{thm:CS weak} immediately gives the following corollary with better dependency on $\delta$.
\begin{corollary}
For any $\epsilon>0$, $\delta\in(0,1)$, and $n\in\N$. For any insertion-only stream over $[n]$ with frequency $f^{(1)},f^{(2)},\dots,f^{(m)}$, there exists a  streaming algorithm providing $\ell_2$ $(\epsilon,\delta)$-weak tracking with $k=O(\epsilon^{-2}\log(1/\delta))$ rows and update time $O(\log(1/\delta))$.
\end{corollary}

The proof of~\autoref{thm:CS weak} uses the Dudley-like chaining technique
similar to other tracking proofs~\cite{braverman2017bptree}. However, direct application of the chaining argument would not suffice and we have to utilize the structure of the sketching matrix of \textsf{CountSketch} (see~\autoref{rmk:differences} for comparison). We will prove~\autoref{thm:CS weak} in~\autoref{sec:proof of CS weak}.

\subsection{Proof of~\autoref{thm:CS weak}}\label{sec:proof of CS weak}
In this subsection, we give a formal proof for our main theorem.
Let us start with some notations for \textsf{CountSketch}. Recall that for any $i\in[n]$, the $i^{\text{th}}$ column of $\Pi$ is defined by (i) picking $j\in[k]$ uniformly and set $\Pi_{j,i}$ to be a Rademacher random variable and (ii) set the other entries in $\Pi_i$ to be 0. Denote $\Pi_{j,i} = \sigma_{j,i}\eta_{j,i}$, where $\sigma_{j,i}$ is a Rademacher random variable, and $\eta_{j,i}$ is the indicator for choosing the $j^{\text{th}}$ row in the $i^{\text{th}}$ column. Note that there is exactly one non-zero entry in each column and the probability distribution is uniform. The approximation error of $\Pi$ for a vector $\bx\in\R^n$ is denoted as $\gamma(\bx):=\left|\|\Pi \bx\|_2^2-\|\bx\|_2^2\right|$. To show weak tracking, it suffices to upper bound the supremum of $\gamma(f^{(t)})$.
\begin{equation}\label{eq:supremum error}
\E_{\Pi}\sup_{t\in[m]}\gamma(f^{(t)}) = \E_{\Pi}\sup_{t\in[m]}\Bigl|\|\Pi f^{(t)}\|_2^2-\|f^{(t)}\|_2^2\Bigr|.
\end{equation}

The first observation\footnote{Note that the matrix $\tB_{\bx}$ we are using is different from the matrix used in the previous analysis of~\cite{braverman2017bptree}. This difference is crucial since the matrix of~\cite{braverman2017bptree} does not work for \textsf{CountSketch}.} is that one can rewrite the error $\gamma(\bx)$ as follows.
\begin{equation*}
\gamma(\bx) = \left|\bx^\top\Pi^\top\Pi \bx-\bx^\top \bx\right| = \left|\sigma^\top B_{\eta,\bx}\sigma - \bx^\top \bx\right|=\left| \sigma^\top\tB_{\eta,\bx}\sigma\right|,
\end{equation*}
where $\sigma\in\{-1,1\}^{n}$ is an independent Rademacher random vector and for any $i,i'\in[n]$,
\begin{equation*}
(\tB_{\eta,\bx})_{i,i'}=\left\{\begin{array}{ll}
\bx_i\bx_{i'},&\text{ $i\neq i'$ and }\exists j\in[k],\ \eta_{j,i}=\eta_{j,i'}=1\\
0,&\text{ else.}
\end{array}
\right.
\end{equation*}
Note that the diagonals of $\tB_{\eta,\bx}$ are all zero as follow.
\[\tB_{\eta,\bx} = \begin{pmatrix}
0&\bx_1\bx_2\langle\Pi_1,\Pi_2\rangle&\cdots&\bx_1\bx_n\langle\Pi_1,\Pi_n\rangle\\
\bx_2\bx_1\langle\Pi_2,\Pi_1\rangle&0&\cdots&\bx_2\bx_n\langle\Pi_2,\Pi_n\rangle\\
\vdots&\vdots&\ddots&\vdots\\
\bx_n\bx_1\langle\Pi_n,\Pi_1\rangle&\bx_n\bx_2\langle\Pi_n,\Pi_2\rangle&\cdots&0
\end{pmatrix}.
\]
For convenience, for any matrix $B\in\R^{n\times n}$, we overload the notation $\gamma$ by denoting $\gamma(B)=\sigma^\top B\sigma$. That is, $\gamma(\tB_{\eta,\bx})=\gamma(\bx)$. One benefit of writing $\ell_2$ weak tracking error into the above quadratic form is that Hanson-Wright inequality (see~\autoref{lem:HW}) is now applicable.

The lemma below shows that the expectation of the weak tracking error is upper bounded by the Frobenius norm of $\tB_{\eta,f^{(m)}}$.

\begin{lemma}\label{lem:bound error with Frobenius}
Let $\{f^{(t)}\}_{t\in[m]}$ be the frequencies of an insertion-only stream. We have
\begin{equation*}
\E\left[\sup_{t\in[m]}\gamma(f^{(t)})\ |\ \eta\right] = O(\|\tB_{\eta,f^{(m)}}\|_F).
\end{equation*}
\end{lemma}

The proof of~\autoref{lem:bound error with Frobenius} uses the Dudley-like chaining argument. For the smooth of presentation, we postpone the details to~\autoref{sec:proof of key lemma}.
Next, the following lemma shows that for any vector $x\in\R^n$, with high probability, $\|\tB_{\eta,x}\|_F=O(\|x\|_2^2/\sqrt{k})$.

\begin{lemma}\label{lem:concentration Frobenius}
For any $\delta\in(0,1)$ and $x\in\R^n$,
\begin{equation*}
    \Pr\left[\|\tB_{\eta,x}\|_F> \frac{\sqrt{2}\|x\|_2^2}{\sqrt{\delta\cdot k}}\right]\leq\frac{\delta}{2}.
\end{equation*}
\end{lemma}

\autoref{lem:concentration Frobenius} has similar flavor as~\autoref{lem:CS estimation}. The proof can be found in~\autoref{sec:proof of key lemma}.
Finally, \autoref{thm:CS weak} is an immediate corollary of~\autoref{lem:bound error with Frobenius} and~\autoref{lem:concentration Frobenius}. Here we provide a proof for completeness.

\begin{proof}[Proof of~\autoref{thm:CS weak}]
Recall that to prove~\autoref{thm:CS weak}, it suffices to show that with probability at least $1-\delta$ over $\eta$, $\sup_{t\in[m]}\gamma(f^{(t)})\leq\eps$. From~\autoref{lem:bound error with Frobenius}, for a fixed $\eta$, we have $\Pr\left[\sup_{t\in[m]}\gamma(f^{(t)}) > C_1\|\tB_{\eta,f^{(m)}}\|_F\right]\leq\delta/2$ for some constant $C_1>0$. Next, from~\autoref{lem:concentration Frobenius}, we have $\|\tB_{\eta,f^{(m)}}\|_F\leq \|f^{(m)}\|_2^2\cdot k^{-1/2}\cdot\delta^{-1/2}$ with probability at least $1-\delta/2$ over the randomness in $\eta$ for some constant $C_2>0$. Pick $m\geq C_1C_2\cdot\epsilon^{-2}\cdot\delta^{-1}$, we have $\Pr\left[\sup_{t\in[m]}\gamma(f^{(t)})>\epsilon\|f^{(m)}\|_2^2\right]\leq\delta$ and complete the proof.
\end{proof}

\subsection{Proof of the two key lemmas}\label{sec:proof of key lemma}
In this subsection, we provide the proofs for~\autoref{lem:bound error with Frobenius} and~\autoref{lem:concentration Frobenius}. Let us start with~\autoref{lem:bound error with Frobenius} which shows that the tracking error can be upper bounded by the Frobenius norm of $\tB_{\eta,f^{(m)}}$.
\begin{proof}[Proof of~\autoref{lem:bound error with Frobenius}]
Recall that we define $\tB_{\eta,x}$ such that $\gamma(x)=\sigma^\top\tB_{\eta,x}\sigma$ where $\sigma$ is 8-wise independent Rademacher random vector. An important trick here is that we think of \textit{fixing}\footnote{We do this by conditioning on $\eta$.} $\eta$ in the following.

The starting point of chaining argument is constructing a sequence of $\epsilon$-nets with exponentially decreasing error for $\{\tB_{\eta,f^{(t)}}\}_{t\in[m]}$.
Note that here $\{\tB_{\eta,f^{(t)}}\}_{t\in[m]}$ are matrices but one can view it as a vector and apply~\autoref{lem:eps net insertion-only stream} where $\ell_2$ norm for a vector becomes Frobenius norm for a matrix. Namely, for any non-negative integer $\ell$, let $T_{\eta,\ell}$ be the $(\|\tB_{\eta,f^{(m)}}\|_F/2^\ell)$-net for $\{\tB_{\eta,f^{(t)}}\}_{t\in[m]}$ under Frobenius norm where $|T_{\eta,\ell}|\leq 1+2^{2\ell}$. Note that here we fixed $\eta$ first and then constructed the nets. Thus, for each $t\in[m]$, one can rewrite $\tB_{\eta,f^{(t)}}$ into a \textit{chain} as follows.
\begin{equation}\label{eq:dudley}
\tB_{\eta,f^{(t)}} = B_{\eta,0}^{(t)} + \sum_{\ell=1}^\infty B_{\eta,\ell}^{(t)}-B_{\eta,\ell-1}^{(t)},
\end{equation}
where $B_{\eta,\ell}^{(t)}\in T_{\eta,\ell}$ and $\|\tB_{\eta,f^{(t)}}-B_{\eta,\ell}^{(t)}\|_F\leq2^{-\ell}\cdot\|\tB_{\eta,f^{(m)}}\|_F$. Moreover, from~\autoref{eq:dudley} we have
\begin{equation}\label{eq:CS weak chaining}
\E\sup_{t\in[m]}\gamma(f^{(t)}) \leq \E\sup_{t\in[m]}\gamma(B_{\eta,0}^{(t)}) + \sum_{\ell=1}^\infty\E\sup_{t\in[m]}\gamma(B_{\eta,\ell}^{(t)}-B_{\eta,\ell-1}^{(t)}).
\end{equation}

To bound to first term of~\autoref{eq:CS weak chaining}, observe that $T_{\eta,0}=\{\tB_{\eta,f^{(1)}}\}$ where $\tB_{\eta,f^{(1)}}$ is the all zero matrix. Namely, the first term of~\autoref{eq:CS weak chaining} is zero. As for the second term of~\autoref{eq:CS weak chaining}, we apply the chaining argument as follows. For any positive integer $\ell$, denote $\mathcal{A}_{\ell}=\{B^{(t)}_{\eta,\ell}-B^{(t)}_{\eta,\ell-1}\}_{t\in[m]}$. Note that from the construction of $\epsilon$-net in~\autoref{lem:eps net insertion-only stream}, we have $|\mathcal{A}_{\ell}|\leq2|T_{\eta,\ell}|\leq2^{2\ell+2}$ by triangle inequality.

\begin{align}
\E\left[\sup_{t\in[m]}\gamma(B^{(t)}_{\eta,\ell}-B^{(t)}_{\eta,\ell-1})\right] &= \int_0^\infty\Pr\left[\sup_{A\in\mathcal{A}_\ell}\gamma(A)>u\right]du\nonumber\\
&\leq u_\ell^* + \int_{u_\ell^*}^\infty \Pr\left[\sup_{A\in\mathcal{A}_\ell}\gamma(A)>u\right]du,\label{eq:main proof 1}
\end{align}
where $u_\ell^*>0$ will be chosen later. For any $A\in\mathcal{A}_\ell$ and integer $p\geq2$, by Markov's inequality and Hanson-Wright inequality, we have

\begin{equation*}
\Pr[\gamma(A)>u] \leq \frac{\E[\gamma(A)^p]}{u^p} =\frac{\|\sigma^\top A\sigma\|_p^p}{u^p}\leq \frac{\left(C\cdot\sqrt{p}\|A\|_F+C\cdot p\|A\|\right)^p}{u^p}
\end{equation*}
for some constant $C>0$. Note that the randomness here is only in $\sigma$ and thus we can apply the Hanson-Wright inequality. Let \fix{$R_\ell=\sup_{A\in\mathcal{A}_\ell}\left(C\cdot\sqrt{p}\|A\|_F+C\cdot p\|A\|\right)\leq C'p\cdot\|\tB_{\eta,f^{(m)}}\|_F\cdot2^{-\ell}$}~\footnote{In the submission, we didn't treat $R_\ell$ differently for each $\ell$.} for some $C'>0$. The last inequality holds because of $\|\cdot\|\leq\|\cdot\|_F$ and the choice of $\eps$-net. Now, choose \fix{$u_\ell^*=2S_\ell\cdot R_\ell$} where $S_\ell$ will be decided later,~\autoref{eq:main proof 1} becomes

\begin{align}
\E\left[\sup_{t\in[m]}\gamma(B^{(t)}_{\eta,\ell}-B^{(t)}_{\eta,\ell-1})\right] &\leq u_\ell^* + \int_{u_\ell^*}^\infty |\mathcal{A}_\ell| \cdot \frac{\fix{R_\ell}^p}{u^p} du\label{eq:CS weak union}\\
&\leq 2S_\ell \fix{R_\ell} + |\mathcal{A}_\ell|\cdot\frac{\fix{R_\ell}^p}{(2S_\ell \fix{R_\ell})^{p-1}}\nonumber\\
&\leq 2S_\ell C'p\cdot\|\tB_{\eta,f^{(m)}}\|_F\fix{\cdot2^{-\ell}} + |\mathcal{A}_\ell|\cdot\frac{C'p\cdot\|\tB_{\eta,f^{(m)}}\|_F}{S_\ell^{p-1}}\fix{\cdot2^{-\ell}}\nonumber
\end{align}
where the second term of~\autoref{eq:CS weak union} is due to union bound. Now,~\autoref{eq:CS weak chaining} becomes

\begin{align}
\E\sup_{t\in[m]}\gamma(f^{(t)}) &\leq\sum_{\ell=1}^\infty 2S_\ell C'p\cdot\|\tB_{\eta,f^{(m)}}\|_F\fix{\cdot2^{-\ell}} + |\mathcal{A}_\ell|\cdot\frac{C'p\cdot\|\tB_{\eta,f^{(m)}}\|_F}{S_\ell^{p-1}}\fix{\cdot2^{-\ell}}\nonumber\\
&\leq\|\tB_{\eta,f^{(m)}}\|_F\cdot\left(\sum_{\ell=1}^\infty2C'pS_\ell\cdot2^{-\ell}+\frac{2^{\fix{\ell}}C'p}{S_\ell^{p-1}}\right).\label{eq:main proof 2}
\end{align}

Choose $S_\ell=2^{3\ell/4}$ and $p\geq4$, the summation term in~\autoref{eq:main proof 2} can thus be upper bounded by a constant. We conclude that
\begin{equation*}
\E\sup_{t\in[m]}\gamma(f^{(t)}) = O(\|\tB_{\eta,f^{(m)}}\|_F).
\end{equation*}
Note that this also means that 8-wise independence suffices and thus the sketching matrix can be efficiently stored (see~\autoref{sec:CS} for more details).
\end{proof}

Next, we prove~\autoref{lem:concentration Frobenius} which upper bounds the expectation of $\|\tB_{\eta,\bx}\|$ for any $\bx\in\R^n$.

\begin{proof}[Proof of~\autoref{lem:concentration Frobenius}]
We first show that $\E_{\eta}\|\tB_{\eta,x}\|_F^2\leq\frac{\|x\|_2^4}{k}$ and the lemma immediately holds due to Markov's inequality.

Let $\mathbf{1}_{ii'}$ be the indicator for whether there exists $j\in[k]$ such that $\eta_{ij}=\eta_{i'j}=1$. Note that for $i\neq i'$, $\E[\mathbf{1}_{ii'}]=1/k$ and the only randomness here is in $\eta$.
\begin{align*}
\E\|\tB_{\eta,x}\|_F^2 &= \E\sum_{i,i'\in[n]}(\tB_{\eta,x})_{i,i'}^2=\E\sum_{(i,i')\in[n]^2,\ i\neq i'}x_i^2x_{i'}^2\mathbf{1}_{ii'}\\
&=\frac{1}{k}\sum_{(i,i')\in[n]^2,\ i\neq i'}x_i^2x_{i'}^2\leq\frac{\|x\|_2^4}{k},
\end{align*}
where the last inequality is by Cauchy-Schwarz. Note that 8-wise independence is sufficient in the above argument.
\end{proof}

\begin{remark}\label{rmk:differences}
Here, let us briefly compare the difference between our techniques and that in~\cite{braverman2017bptree}. There are two key observations on the structure of the sketching matrix of \textsf{CountSketch}.
First, we observe that the Frobenius norm of $\Pi^\top\Pi$ is dominated by its diagonal and thus \textit{removing} the diagonal would give us a more accurate analysis on the contribution from the off-diagonal term. However, removing the diagonal of $\Pi^\top\Pi$ destroys the symmetric structure and thus the standard $\eps$-net argument (e.g., in~\cite{braverman2017bptree}) would not work. To overcome this, we observe that one can directly construct $\eps$-net for the matrix obtained by removing the diagonal from $\Pi^\top\Pi$. Combining these two observations and standard chaining argument, we are able to show that \textsf{CountSketch} provides $\ell_2$ weak tracking.
\end{remark}

\section{Strong tracking of \textsf{AMS} sketch and \textsf{CountSketch}}\label{sec:strong}
In this section, we are going to discuss the strong tracking of \textsf{AMS} sketch and \textsf{CountSketch}. We start with a standard reduction from weak tracking to strong tracking via union bound. This gives us an $O(\log m)$ blow-up in the dependency on $\delta$. Next, we show that this is essentially tight for both $\textsf{AMS}$ sketch and \textsf{CountSketch} up to a logarithmic factor.

\begin{lemma}[folklore]\label{lem:weak to strong}
For any $\epsilon>0$, $\delta\in(0,1)$, and $n,m\in\N$. If a linear sketch provides $(\epsilon,\delta)$ weak tracking for length $m$ inputs having value from $[n]$, then it also provides $(2\epsilon,\delta')$ strong tracking where $\delta'=\min\{1,(\log m)\cdot\delta\}$.
\end{lemma}
\begin{proof}
See~\autoref{sec:proof weak to strong} for details.
\end{proof}

From~\autoref{lem:weak to strong}, we immediate have the following corollaries.
\begin{corollary}
For any $\epsilon>0$ and $\delta\in(0,1)$, \textsf{AMS} sketch with $O\left(\epsilon^{-2}(\log\log m + \log(1/\delta))\right)$ rows provides $\ell_2$ $(\epsilon,\delta)$-strong tracking.
\end{corollary}

\begin{corollary}
For any $\epsilon>0$ and $\delta\in(0,1)$, \textsf{CountSketch} with $O\left(\epsilon^{-2}\delta^{-1}\log m\right)$ rows provides $\ell_2$ $(\epsilon,\delta)$-strong tracking.
\end{corollary}

\begin{remark*}
After applying median trick on \textsf{CountSketch}, the dependency of the number of rows on $\delta$ becomes $O(\log(1/\delta))$ and thus $O\left(\epsilon^{-2}(\log\log m + \log(1/\delta))\right)$ rows suffices to achieve $\ell_2$ $(\epsilon,\delta)$-strong tracking.
\end{remark*}

In the following, we are going to show that the above two upper bounds are essentially tight for these two algorithms.

\begin{theorem}\label{thm:strong AMS lower bound}
There exists constants $C>0$ such that for any $\epsilon\in(0,0.1)$ and $\delta\in(0,1)$, there exists $N_0\in\N$ such that if $k<C\cdot\left(\log\frac{\log m}{\log(1/\epsilon)}+\log(1/\delta)\right)$ and $N_0\leq n\leq m$, then fully independent \textsf{AMS} sketch with $k$ rows does not provide $\ell_2$ $(\epsilon,\delta)$-strong tracking.
\end{theorem}
That is, $\textsf{AMS}$ sketch requires $\tilde{\Omega}\left(\epsilon^{-2}(\log\log m+\log(1/\delta))\right)$ rows to achieve $\ell_2$ $(\epsilon,\delta)$-strong tracking. Interestingly, the hard instance for \textsf{AMS} sketch to achieve strong tracking is simply the stream consisting all distinct elements. See~\autoref{sec:proof strong AMS} for details.
 
\begin{theorem}\label{thm:CS no strong}
There exists a constant $C>0$ such that for any $\epsilon\in(0,0.5)$, and $\delta\in(0,1)$, there exists $N_0\in\N$ such that if $k\leq C\cdot\epsilon^{-2}\delta^{-1}\frac{\log m}{\log(1/\epsilon)}$ and $N_0\leq n\leq O(\log m)$, then \textsf{CountSketch} with $k$ rows does not provide $\ell_2$ $(\epsilon,\delta)$-strong tracking.
\end{theorem}
That is, \textsf{CountSketch} requires $\tilde{\Omega}(\epsilon^{-2}\delta^{-1}\log m)$ rows to achieve $\ell_2$ $(\epsilon,\delta)$-strong tracking. The hard instance for \textsf{CountSketch} is more complicated than that of \textsf{AMS} sketch. See~\autoref{sec:proof strong CS} for details.

\section{Conclusion}\label{sec:conclusion}
In this work, we showed that \textsf{CountSketch} provides $\ell_2$ weak tracking with update time having no dependence on the error parameter $\epsilon$. We also give almost tight $\ell_2$ strong tracking lower bounds for \textsf{AMS} sketch and \textsf{CountSketch}.

An immediate open problem after this work would be tracking $\ell_p$ with faster update time for $0<p<2$.
The $\ell_p$ estimation problem had been solved by Indyk~\cite{indyk2006stable} via \textit{$p$-stable sketch} and was proven to provide weak tracking by B\l{}asiok et al.~\cite{blasiok2017continuous}. However, same as \textsf{AMS} sketch, the $p$-stable sketch is dense and has update time $\Omega(\epsilon^{-2})$. Nevertheless, Kane et al.~\cite{KNPW11} gave a space-optimal algorithm for $\ell_p$ estimation problem with update time $O(\log^2(1/\epsilon)\log\log(1/\epsilon))$. It would be interesting to see if their algorithm also provides $\ell_p$ weak tracking.

\subsection*{Acknowledgement}
The authors wish to thank Jelani Nelson for invaluable advice throughout the course of this research.
We also thank Mitali Bafna and Jaros\l{}aw B\l{}asiok for useful discussion and thank Boaz Barak for many helpful comments on an earlier draft of this article.
\bibliographystyle{alpha}
\bibliography{mybib}

\appendix

\section{Implementation of \textsf{CountSketch}}\label{sec:CS}
Here, we present the implementation of \textsf{CountSketch} for the completeness. Note that the construction is standard and not new.

\begin{algorithm}[H] 
	\caption{Constructing \textsf{CountSketch}}\label{algo:count sketch}
	\begin{algorithmic}[1] 
		\State $k\leftarrow\ceil*{\frac{c}{\epsilon^2}}$ for some constant $c>0$.
		\State $\tilde{f}\in\Z^k$ vector with initial value 0.
		\State Sample $h:[n]\rightarrow[k]$ from a 8-wise independent hash family.
		\State Sample $g:[n]\rightarrow\{\pm1\}$ from a 8-wise independent hash family.
		\For{$t=1,2,\dots,m$}
		\State On input $a_t=i$, set $\tilde{f}_{h(i)}=\tilde{f}_{h(i)}+g(i)$.
		\EndFor
	\end{algorithmic}
\end{algorithm}

Note that both $h$ and $g$ can be stored in space $O(\log n+\log(1/\epsilon))$ and be evaluated in $O(1)$ many arithmetic operations. $\tilde{f}$ can be stored in space $O(\epsilon^{-2}\log m)$ bits. For the convenience of analysis, we define the sketching matrix $\Pi\in\{0,\pm1\}^{k\times n}$ of \textsf{CountSketch} by $\Pi_{h(i),i}=g(i)$ for all $i\in[n]$.

\section{Proofs for strong tracking}\label{sec:proof of strong tracking}

\subsection{From weak tracking to strong tracking}\label{sec:proof weak to strong}
After applying union bound on all points $t=1,2,\dots,m$, a streaming algorithm provides $\ell_2$ $(\epsilon,\delta)$-approximation also provides $\ell_2$ $(\epsilon,\delta')$-strong tracking where $\delta'=\min\{1,m\delta\}$. However, the blow-up in $\delta$ is $m$, which is undesirable. The following lemma shows that with a more delicate union bound argument, the reduction from weak tracking to strong tracking only has $O(\log m)$ blow-up in $\delta$. Note that the lemma is a folklore and we provide a proof for completeness.

\begin{proof}
	Let $\{f^{(t)}\}_{t\in[m]}$ be the frequency of an insertion-only stream and let $\{\tilde{f}^{(t)}\}_{t\in[m]}$ be its (randomized) approximations produced by the linear sketch. Let $w=\floor*{\log m}+1$ and $t_i=2^i-1$ for each $i\in[w]$. Note that for each $i\in[w]$ and $t_{i-1}< t\leq t_i$, $\frac{1}{2}\|f^{(t_i)}\|_2^2\leq\|f^{(t)}\|_2^2\leq\|f^{(t_i)}\|_2^2$. Define the event
	\begin{equation*}
	E_i:=\left\{\|\tilde{f}^{(t_i)}\|_2^2-\|f^{(t_i)}\|_2^2|>\epsilon\|f^{(t_i)}\|_2^2 \right\}.
	\end{equation*}
	Observe that for each $t_{i-1}<t\leq t_i$, $|\|\tilde{f}^{(t)}\|_2^2-\|f^{(t_i)}\|_2^2|>2\epsilon\cdot\|f^{(t)}\|_2^2$ would imply $\neg E_i$. Namely, $\neg\cup_{i\in[w]}E_i$ implies strong tracking.
	
	By the $\ell_2$ $(\epsilon,\delta)$-weak tracking property of the streaming algorithm, for each $i\in [w]$, we have $\Pr\left[E_i\right]\leq\delta$ and thus $\Pr[\cup_{i\in[w]}E_i]\leq w\delta$.
	We conclude that the streaming algorithm provides $\ell_2$ $(2\epsilon,w\delta)$-strong tracking.
\end{proof}

\subsection{Strong tracking lower bound for \textsf{AMS} sketch}\label{sec:proof strong AMS}
The hard instance is simply the stream of all distinct elements, \textit{i.e.,} $i_t=t$ for all $t\in[m]$.
\begin{proof}[Proof of~\autoref{thm:strong AMS lower bound}]
	Consider the stream of all distinct elements as the hard instance, \textit{i.e.,} $i_t=t$ for all $t\in[m]$. Thus, $\|f^{(t)}\|_2^2=t$ and $\|\Pi f^{(t)}\|_2^2=\sum_{i\in[k]}\left(\sum_{j\in[t]}\Pi_{i,j}\right)^2$ for all $t\in[m]$.
	
	Define a sequence of time $\{t_j\}$ as follows. $t_0=0$ and $t_j=\sum_{i\in[j]}\Delta_i$ where $\Delta_i=\ceil*{10/\epsilon}^i$. Pick $\ell$ and $m$ properly such that $t_\ell\leq m$. Some quick facts about the choice of parameters here: (i) $|t_j-\Delta_j|\leq\frac{\epsilon}{5}\cdot t_j$. (ii) $\ell=\Theta(\frac{\log m}{\log(1/\epsilon)})$.
	
	To show \textsf{AMS} sketch does not provide $(\epsilon,\delta)$-strong tracking for $\epsilon\in(0,0.1)$ and $\delta\in(0,1)$, it suffices to show that with probability at least $\delta$ there exists $j\in[\ell]$ such that $\|\Pi f^{(t_j)}\|_2^2-t_j>(1+\epsilon)\cdot t_j$. 
	
	For the convenience of the analysis, for any $i\in[k]$ and $j\in[\ell]$, let $X^{(t_j)}_i=\sum_{s=t_{j-1}+1}^{t_j}\Pi_{i,s}$ which is the sum of $\Delta_j$ independent Rademacher random variables divided by $\sqrt{k}$. Also let $Z_j=\sum_{i\in[k]}(X^{(t_j)}_i)^2$.
	Note that $\E[Z_j]=\Delta_j/\sqrt{k}$ and
	\begin{align}
	\|\Pi f^{(t_j)}\|_2^2&=\sum_{i\in[k]}\left(\sum_{j'\in[j]}X^{(t_{j'})}_i\right)^2\nonumber\\
	&=Z_j + \sum_{i\in[k]}\left(\sum_{j'\in[j-1]}X^{(t_{j'})}_i\right)^2+2\sum_{i\in[k]}\langle X^{(t_j)}_i,\sum_{j'\in[j-1]}X^{(t_{j'})}_i \rangle.\label{eq:strong AMS lb 1}
	\end{align}
	Define an event $E_j:=\{Z_j\geq(1+2\epsilon)\cdot\E[Z_j]\}$ for each $j\in[\ell]$. Observe that when conditioning on $\cap_{j'\in[j-1]}\neg E_{j'}$, the second term of~\autoref{eq:strong AMS lb 1} is bounded by $O(t_{j-1})$ and the third term is bounded by $O(\sqrt{t_{j-1}Z_j})$ due to Cauchy-Schwarz. By the choice of parameters, both term can be bounded by $0.1t_j$. Furthermore, $E_j$ implies $\|\Pi f^{(t_j)}\|_2^2-t_j>(1+\epsilon)\cdot t_j$. Note that $E_j$ is independent to $E_1,\dots,E_{j-1}$. The following lemma lower bound the probability of $E_j$ to happen.
	
	\begin{lemma}\label{lem:strong AMS lb anticoncentration}
		There exists a constant $c>0$ such that $\Pr[E_j]\geq e^{-c\epsilon^2k}$ for any $j=\Omega(\log\log k)$.
	\end{lemma}
	\begin{proof}[Proof of~\autoref{lem:strong AMS lb anticoncentration}]
		From the seminal \textit{Berry-Esseen theorem}~\cite{berry1941accuracy,esseen1942liapounoff}, we know that when $t_j=e^{\Omega(k)}=\Omega(\frac{\log m}{\delta})$ then $X^{(t_j)}$ is point-wisely $e^{-\Omega(k)}$-close to a normal distribution with zero mean and variance $\Delta_j$. That is, $\frac{kZ_j}{\Delta_j}$ is also point-wisely $e^{-\Omega(k)}$-close to a \textit{chi-square} distribution $\chi^2_{\Delta_j}$ with mean $\Delta_j$ and $\Delta_j$ degree of freedom\footnote{Recall that a \textit{chi-square random variable} of $d$ degree of freedom is equivalent to the sum of $d$ squares of the standard normal random variable.}.
		
		Inglot and Ledwina~\cite{inglot2006asymptotic} showed that the tail of chi-square random distribution can be lower bounded as $\Pr[\chi^2_k\geq(1+2\epsilon)\cdot k]\geq\frac{1}{2}e^{-\epsilon^2k/10}$ when $k$ large enough.
		Combine with the Berry-Esseen theorem, we have $\Pr[E_j]\geq e^{-c\epsilon^2k}$ for some constant $c>0$.
	\end{proof}
	
	Note that as $\{Z_j\}_{j\in[\ell]}$ are mutually independent, the events $\{E_j\}_{j\in[\ell]}$ are also mutually independent. That is,
	\begin{align*}
	\Pr\left[\exists t\in[m],\ \left|\|\Pi f^{(t)}\|_2^2-\|f^{(t)}\|_2^2\right|>2\epsilon\|f^{(t)}\|_2^2\right]&\geq\Pr\left[\cup_{ j\in[\ell]}E_j\right]\\
	&\geq1-\prod_{j\in[\ell]}\Pr\left[\neg E_j\ | \ \neg E_{j'},\ \forall j'\in[j-1]\right]\\
	&\geq1-\left(1-e^{-c\epsilon^2k}\right)^\ell\geq\ell e^{-c\epsilon^2k}.
	\end{align*}
\end{proof}
Namely, there exists another constant $C>0$ such that if $k<C\epsilon^{-2}\left(\log\frac{\log m}{\log(1/\epsilon)} + \log(1/\delta)\right)\leq\frac{1}{c}\epsilon^{-2}\log \frac{\ell}{\delta}$. Thus, \textsf{AMS} sketch does not provide $(\epsilon,\delta)$-strong tracking for all $\epsilon\in(0,0.1)$.

\subsection{Strong tracking lower bound for \textsf{CountSketch}}\label{sec:proof strong CS}
To prove~\autoref{thm:CS no strong}, we are going to construct a stream such that any \textsf{CountSketch} does not provide strong tracking. Let's start from some observation. For any $i\neq i'\in[n]$ and $a>0$, let $\bx=a(\be_i+\be_{i'})$ such that $\|\bx\|_2^2=2a^2$. Now, observe that If $\Pi_i=\Pi_{i'}$, then we have $\|\Pi \bx\|_2^2=4a^2$. If $\Pi_i=-\Pi_{i'}$, then we have $\|\Pi \bx\|_2^2=0$. Note that in both cases, the approximation $\|\Pi \bx\|_2^2$ and the correct answer $\|\bx\|_2^2$ has a huge gap $2a^2$, \textit{i.e.,} $\left|\|\Pi \bx\|_2^2-\|\bx\|_2^2 \right|\geq\|\bx\|_2^2$.

With the above observation, one can see that a collision (either $\Pi_i=\Pi_{i'}$ or $\Pi_i=-\Pi_{i'}$) is a sufficient condition for an estimation error. As a result, to show \textsf{CountSketch} does not provide strong tracking, it suffices to show the following two things: (i) there will be some collision with constant probability and (ii) construct a stream such that once a collision happens, the estimation error is large.

Note that (ii) is very specific to tracking since unlike $\ell_2$ estimation which only cares about the final estimation, we need to keep track of the estimation at any time. Thus, to show the impossibility of tracking, we have to show that the estimation fails at least once at some point.

\begin{proof}[Proof of~\autoref{thm:CS no strong}]
	Let $n$ be the number of elements and $k$ be the number of rows of \textsf{CountSketch}. Let $\Delta=\ceil*{100/\epsilon}$ and $w=\ceil*{1/\epsilon}$.
	For any $j\in[\ell]$, define $t_j=\sum_{j'\in[j]}\Delta^{j'+1}=\frac{\Delta^{j+1}-\Delta^1}{\Delta-1}$ and the stream at time $t_j$ as follows.
	\begin{equation*}
	f^{(t_j)} = \left(\underbrace{\Delta, \dots, \Delta}_{w}, \underbrace{\Delta^2, \dots, \Delta^2}_{w}, \underbrace{\Delta^j, \dots, \Delta^j}_{w},0,\dots,0\right).
	\end{equation*}
	We have $\|f^{(t_j)}\|_2^2=\sum_{j'\in[j]}w\cdot\Delta^{2j'+1}=\frac{w\cdot\Delta^{2j+2}-w\cdot\Delta^2}{\Delta^2-1}$.
	Note that one can easily complete rest of the stream $\{f^{(t)}\}_{t\in[m]}$ for any $m\geq t_\ell$. Note that here we can pick $\ell=\Theta(\frac{\log m}{\log(1/\epsilon)})$.
	
	Define the event $E_j:=\{\|\Pi f^{(t_j)}\|_2^2-\|f^{(t_j)}\|_2^2>\epsilon\cdot\|f^{(t_j)}\|_2^2\}$. To show that \textsc{CountSketch} does not provide $w_2$ $(\epsilon,\delta)$-strong tracking, it suffices to prove $\Pr[\cup_{j\in[\ell]}E_j]>\delta$. The following lemma lower bounds the probability of single $E_j$.
	
	\begin{lemma}\label{lem:strong CS 1}
		For each $j\in\ell$, we have $\Pr[E_j\ |\ \neg\cup_{j'\in[j]}E_{j'}]\geq\frac{1}{10k\epsilon^2}$.
	\end{lemma}
	\begin{proof}
		First, let $\bar{f}^{(t_j)} = f^{(t_j)} - f^{(t_{j-1})}$ for each $j\in\ell$ where we define $f^{(0)}=\mathbf{0}$. Observe that
		\begin{align*}
		\|\Pi f^{(t_j)}\|_2^2-\|f^{(t_j)}\|_2^2 &= \|\Pi \bar{f}^{(t_j)} + \Pi f^{(t_{j-1})}\|_2^2 - \|\bar{f}^{(t_j)} + f^{(t_{j-1})}\|_2^2\\
		&= \|\Pi\bar{f}^{(t_j)}\|_2^2 - \|\bar{f}^{(t_j)}\|_2^2 + \|\Pi f^{(t_{j-1})}\|_2^2 - \|f^{(t_{j-1})}\|_2^2\\
		&+ 2\langle\Pi\bar{f}^{(t_j)},\Pi f^{(t_{j-1})}\rangle - 2\langle\bar{f}^{(t_j)},f^{(t_{j-1})}\rangle.
		\end{align*}
		Further, condition on $\neg\cup_{j'\in[j-1]}E_{j'}$, we have $\|f^{(t_{j-1})}\|_2^2$, $\|\Pi f^{(t_{j-1})}\|_2^2$,  $|\langle\Pi\bar{f}^{(t_j)},\Pi f^{(t_{j-1})}\rangle|$, and $|\langle\bar{f}^{(t_j)},f^{(t_{j-1})}\rangle|$ are all at most $(\epsilon/10)\cdot\|f^{(t_j)}\|_2^2$ by the choice of $\Delta$. Namely,
		\begin{equation}\label{eq:strong CS 1}
		\|\Pi f^{(t_j)}\|_2^2-\|f^{(t_j)}\|_2^2 \geq \|\Pi \bar{f}^{(t_j)}\|_2^2 - \|\bar{f}^{(t_j)}\|_2^2 - \frac{\epsilon}{2}\cdot\|f^{(t_j)}\|_2^2.
		\end{equation}
		\begin{lemma}\label{lem:strong CS 2}
			$\Pr\left[\|\Pi \bar{f}^{(t_j)}\|_2^2 - \|\bar{f}^{(t_j)}\|_2^2>3\epsilon\cdot\|f^{(t_j)}\|_2^2\right]>\frac{1}{10k\epsilon^2}$.
		\end{lemma}
		\begin{proof}
			Let us consider the columns of $\Pi$ that correspond to the non-zero entries of $\bar{f}^{(t_j)}$. That is, column $\Delta\cdot(j-1)+1$ to $\Delta\cdot j$.
			Note that once there are exactly one collision happens among these columns and the both the value are the same, then $\|\Pi\bar{f}^{(t_j)}\|_2^2 - \|\bar{f}^{(t_j)}\|_2^2>3\epsilon\cdot\|f^{(t_j)}\|_2^2$. The probability of the above to happen is at least the following.
			\begin{equation*}
			\frac{1}{2}\cdot\frac{k\cdot\binom{w}{2}\cdot(k-1)\cdot(k-2)\cdots(k-w+2)}{k^w}\geq\frac{w^2}{5k}>\frac{1}{10k\epsilon^2}.
			\end{equation*}
		\end{proof}
		Now,~\autoref{lem:strong CS 1} immediately follows from~\autoref{eq:strong CS 1} and~\autoref{lem:strong CS 2}.
	\end{proof}
	Let us wrap up the proof of~\autoref{thm:CS no strong} as follows.
	\begin{align*}
	\Pr\left[\exists t\in[m],\ \left|\|\Pi f^{(t)}\|_2^2-\|f^{(t)}\|_2^2>\epsilon\|f^{(t)}\|_2^2 \right|\right] &\geq\Pr\left[\cup_{j\in[\ell]}E_j\right]\\
	&=\prod_{j\in[\ell]}\Pr\left[E_j\ |\ \neg\cup_{j'\in[j-1]}E_{j'}\right]\\
	&\geq\left(1-\frac{1}{10k\epsilon^2}\right)^{\ell}\geq1-\frac{\ell}{k\epsilon^2}.
	\end{align*}
	By the choice of parameters, the last quantity would be greater than $\delta$ and thus \textsc{CountSketch} with $k\leq C\cdot\epsilon^{-2}\delta^{-1}\frac{\log(m)}{\log(1/\epsilon)}$ rows does not provide $\ell_2$ $(\epsilon,\delta)$-strong tracking.
\end{proof}

\end{document}